\def\draft{0}  

\documentclass[envcountsect,envcountsame,runningheads]{llncs}

\spnewtheorem{thm}{Theorem}[section]{\bfseries}{\itshape}

\usepackage{amsmath}
\usepackage{amssymb}
\usepackage[shortlabels]{enumitem}
\usepackage{graphicx}
\usepackage{hyperref}
\usepackage{grffile}
\usepackage{comment}
\usepackage{color}
\usepackage{float}
\usepackage{thm-restate}
\usepackage{bm}
\usepackage{mathtools}
\usepackage{float}
\usepackage{subfigure}
\usepackage{graphicx}
\usepackage{ifthen}
\usepackage{xcolor}

\newcommand{\mechanism}{\mathcal{M}}
\newcommand{\analyst}{\mathcal{A}}
\newcommand{\adversary}{\mathcal{A}}
\newcommand{\X}{ \mathcal{X} }
\newcommand{\Y}{ \mathcal{Y} }
\newcommand{\RR}{\mathrm{RR}}

\newcommand{\eps}{\varepsilon}
\newcommand{\analystq}{\mathcal{A}_{q_1,a_1}}
\newcommand{\mechq}[1]{\mechanism^{>1}|_{q_1, #1}}
\newcommand{\epszero}{\eps_0}
\newcommand{\epszerogo}{\eps_0^{>1}}

\newcommand{\concomp}{\mathrm{ConComp}}
\newcommand{\comp}{\mathrm{Comp}}

\newcommand{\halt}{\texttt{halt}}
\newcommand{\View}{\texttt{View}}
\newcommand{\view}[2]{ \View \langle #1, #2 \rangle }
\newcommand{\range}{\mathrm{Range}}

\newcommand{\order}{\texttt{order}}

\newcommand{\postprocessing}{\texttt{Post}}

\newcommand{\at}{\bm{a}^{(t)}}
\newcommand{\newdelta}{r\delta^{\frac{1}{r}}}

\newcommand{\add}[1]{\textcolor{blue}{#1}}
\ifnum\draft=1
\newcommand{\salil}[1]{\textcolor{green}{[Salil: #1]}}
\newcommand{\tianhao}[1]{\textcolor{orange}{[Tianhao: #1]}}
\else
\newcommand{\salil}[1]{}
\newcommand{\tianhao}[1]{}
\fi

\begin{document}
\title{Concurrent Composition of Differential Privacy 
\ifnum\draft=1\\
\normalsize{working draft: please do not distribute}\fi}
%
%
\author{Salil Vadhan\thanks{Supported by NSF grant CNS-1565387, a grant from the Sloan Foundation, and a Simons Investigator Award.} \and Tianhao Wang\thanks{Work done while at Harvard University.}}
\authorrunning{S. Vadhan and T. Wang}
%
\institute{Harvard University\\
\email{salil\_vadhan@harvard.edu}
\and
Princeton University \\
\email{tianhaowang@princeton.edu}}
\maketitle              
\begin{abstract}
We initiate a study of the composition properties of {\em interactive} differentially private mechanisms.  An interactive differentially private mechanism is an algorithm that allows an analyst to adaptively ask queries about a sensitive dataset, with the property that an adversarial analyst's view of the interaction is approximately the same regardless of whether or not any individual's data is in the dataset.  Previous studies of composition of differential privacy have focused on non-interactive algorithms, but interactive mechanisms are needed to capture many of the intended applications of differential privacy and a number of the important differentially private primitives.  

We focus on {\em concurrent composition}, where an adversary can arbitrarily interleave its queries to several differentially private mechanisms, which may be feasible when differentially private query systems are deployed in practice.  We prove that when the interactive mechanisms being composed are {\em pure} differentially private, their concurrent composition achieves 
privacy parameters (with respect to pure or approximate differential privacy) that match the (optimal) composition theorem for noninteractive differential privacy.  We also prove a composition theorem for interactive mechanisms that satisfy approximate differential privacy.  That bound is weaker than even the basic (suboptimal) composition theorem for noninteractive differential privacy, and we leave closing the gap as a direction for future research, along
with understanding concurrent composition for other variants of differential privacy.
\end{abstract}

\section{Introduction}
\salil{I accidentally started editing intro.tex, not realizing that was for your thesis and not the submission, so you may want to revert that file to the previous version}

\newcommand{\cX}{\mathcal{X}}
\newcommand{\cY}{\mathcal{Y}}
\newcommand{\MS}{\mathrm{MultiSets}}

\subsection{Differential Privacy}

Differential privacy is a framework for protecting privacy when performing statistical
releases on a dataset with sensitive information about individuals.  
(See the surveys \cite{dwork2014algorithmic,vadhan2017complexity}.)
Specifically, for a differentially private mechanism, the probability distribution of the mechanism's outputs of a dataset should be nearly identical to the distribution of its outputs on the same dataset with any single individual’s data replaced. To formalize this, we call two datasets $x$, $x'$, each multisets over a data universe $\cX$, {\em adjacent} if one can be obtained from the other by adding or removing a single element of $\cX$.

\begin{definition}[Differential Privacy \cite{dwork2006calibrating}] \label{def:DP}
For $\eps, \delta \ge 0$, a randomized algorithm $\mechanism : \MS(\cX)\rightarrow \cY$ is 
{\em $(\eps, \delta)$-differentially private} if for every pair of adjacent datasets $x, x'\in \MS(\cX)$, we have: 
\begin{equation} \label{req:approxDP}
\forall\ T\subseteq \cY\ \Pr[\mechanism(x) \in T] \le e^\eps\cdot \Pr[\mechanism(x') \in T] + \delta
\end{equation}
where the randomness is over the coin flips of the algorithm $\mechanism$. 
\end{definition}

In the practice of differential privacy, we generally view $\eps$ as ``privacy-loss budget'' that is small but non-negligible (e.g. $\eps=0.1$), and we view $\delta$ as cryptographically negligible (e.g. $\delta=2^{-60}$). We refer to the case where $\delta=0$ as \emph{pure differential privacy}, and the case where $\delta >0$ as \emph{approximate differential privacy}. 

\salil{removed post-processing thm from intro, so insert it back into technical sections where it is needed}\tianhao{done}

\subsection{Composition of Differential Privacy}

A crucial property of differential privacy is its behavior under composition. If we run multiple distinct differentially private algorithms on the same dataset, the resulting composed algorithm is also differentially private, with some degradation in the privacy parameters $(\eps, \delta)$. This property is especially important and useful since in practice we rarely want to release only a single statistic about a dataset. Releasing many statistics may require running multiple differentially private algorithms on the same database. Composition is also a very useful tool in algorithm design. In many cases, new differentially private algorithms are created by combining several simpler algorithms. The composition theorems help us analyze the privacy properties of algorithms designed in this way.

Formally, let $\mechanism_0, \mechanism_1, \dots, \mechanism_{k-1}$ be differentially private mechanisms, we define the composition of these mechanisms by independently executing them.
Specifically, we define $\mechanism = \comp(\mechanism_0,\mechanism_1,\ldots,\mechanism_{k-1})$ as follows: \salil{added use of $\comp$ notation here}
$$\mechanism(x) 
= \left(\mechanism_0(x), \dots, \mechanism_{k-1}(x) \right)
$$
where each $\mechanism_i$ is run with independent coin tosses. For example, this is how we might obtain a mechanism answering a $k$-tuple of queries. 

A handful of composition theorems already exist in the literature. The Basic Composition Theorem says that the privacy degrades at most linearly with the number of mechanisms executed. 

\begin{theorem}[Basic Composition \cite{dwork2006our}]
\label{thm:basic-comp}
For every $\eps \ge 0$, $\delta \in [0, 1]$, if $\mechanism_0, \dots, \mechanism_{k-1}$ are each $(\eps, \delta)$-differentially private mechanisms, then their composition $\comp(\mechanism_0, \dots, \mechanism_{k-1})$ is $(k\eps, k\delta)$-differentially private. 
\end{theorem}

Theorem \ref{thm:basic-comp} shows the global privacy degradation is linear in the number of mechanisms in the composition. 
However, if we are willing to tolerate an increase in the $\delta$ term, the privacy parameter $\eps$ only needs to degrade proportionally to $\sqrt{k}$:

\begin{theorem}[Advanced Composition \cite{dwork2010boosting}]
\label{thm:advanced-comp}
For all $\eps \ge 0$, $\delta \in [0, 1]$,
if $\mechanism_0, \dots, \mechanism_{k-1}$ are each $(\eps, \delta)$-differentially private mechanisms and $k < 1/\eps^2$, then for all $\delta' \in (0, 1/2)$, the composition $(\mechanism_0, \dots, \mechanism_{k-1})$ is $ \left(O \left(\sqrt{k \log (1/\delta')} \right) \cdot \eps, k\delta + \delta' \right)$-differentially private. 
\end{theorem}

Theorem \ref{thm:advanced-comp} is an improvement if $\delta' = 2^{-o(k)}$.
However, despite giving an asymptotically correct upper bound for the global privacy parameter, Theorem \ref{thm:advanced-comp} is not exact. 
Kairouz, Oh, and Viswanath~\cite{kairouz2015composition} shows how to compute the optimal bound for composing $k$ mechanisms where all of them are $(\eps, \delta)$-differentially private.  Murtagh and Vadhan~\cite{murtagh2016complexity} further extends the optimal composition for the more general case where the privacy parameters may differ for each algorithm in the composition:
\salil{replace bound in theorem below with version for $\delta_i,\delta>0$}
\salil{replace citation to Murtagh-Vadhan with journal version, also make sure the bound we are stating is the one from the journal version}
\begin{theorem}[Optimal Composition~\cite{kairouz2015composition,murtagh2016complexity}] 
\label{thm:bound-opt}
If $\mechanism_0, \dots, \mechanism_{k-1}$ are each $(\eps_i, \delta_i)$-differentially private, then given any $\delta_g>0$,  $\comp(\mechanism_0, \dots, \mechanism_{k-1})$ is
$(\eps_g,\delta_g)$-differentially private for the
the least value of $\varepsilon_{g} \geq 0$ such that
$$
\frac{1}{\prod_{i=0}^{k-1}\left(1+\mathrm{e}^{\varepsilon_{i}}\right)} 
\sum_{S \subseteq\{0, \ldots, k-1\}} \max 
\left
\{\mathrm{e}^{\sum_{i \in S} \varepsilon_{i}}-\mathrm{e}^{\varepsilon_{g}} \cdot \mathrm{e}^{\sum_{i \notin S} \varepsilon_{i}}, 0
\right\} 
\leq 1-\frac{1-\delta_{g}}{\prod_{i=0}^{k-1}\left(1-
\delta_{i} \right)}
$$
A special case when all $\mechanism_0, \dots, \mechanism_{k-1}$ are $(\eps, \delta)$-differentially private, then privacy parameter is upper bounded by the least value of $\eps_g \ge 0$ such that
$$
\frac{1}{\left(1+\mathrm{e}^{\varepsilon}\right)^k} 
\sum_{i=0}^k {k \choose i} 
\max 
\left
\{\mathrm{e}^{i \eps }-\mathrm{e}^{\varepsilon_{g}} \cdot \mathrm{e}^{(k-i) \varepsilon}, 0
\right\} 
\leq 1-\frac{1-\delta_g}{(1-\delta)^k}
$$
\end{theorem}

\subsection{Interactive Differential Privacy}

The standard treatment of differential privacy, as captured by Definition~\ref{def:DP}, 
refers to a {\em noninteractive} algorithm $\mechanism$ that takes a dataset $x$ as input and
produces a statistical release $\mechanism(x)$, or a batch by taking $\mechanism=\comp(\mechanism_0,\ldots,\mechanism_{k-1})$.
However, in many of the motivating applications
of differential privacy, we don't want to perform all of our releases in one shot, but rather allow analysts to make adaptive queries to a dataset.  Thus, we should view the mechanism $\mechanism$ as a party in a two-party protocol, interacting with a (possibly adversarial) analyst.  

To formalize the concept of interactive DP, we recall one of the standard formalizations of an interactive protocol between two parties $A$ and $B$. We do this by viewing each party as a function, taking its private input, all messages it has received, and the party's random coins, to the party's next message to be sent out. 

\begin{definition}[Interactive protocols]
An \emph{interactive protocol} $(A, B)$ is any pair of functions. The interaction between $A$ with input $x_A$ and $B$ with input $x_B$ is the following random process (denoted $(A(x_A), B(x_B))$):
\begin{enumerate}
    \item Uniformly choose random coins $r_{A}$ and $r_{B}$ (binary strings) for $A$ and $B$, respectively.
    \item Repeat the following for $i=0, 1, \ldots$: 
    \begin{enumerate}
        \item If $i$ is even, let $m_{i}=A\left(x_A, m_{1}, m_3, \ldots, m_{i-1} ; r_{A} \right)$.
        \item If $i$ is odd, let $m_{i}=B\left(x_B, m_{0}, m_2, \ldots, m_{i-1} ; r_{B} \right)$.
        \item If $m_{i-1} = \halt$, then exit loop.
    \end{enumerate}
\end{enumerate}
\end{definition}

We further define the view of a party in an interactive protocol to capture everything the party ``sees'' during the execution: 
\begin{definition}[View of a party in an interactive protocol] \label{def:view}
Let $(A, B)$ be an interactive protocol. Let $r_{A}$ and $r_{B}$ be the random coins for $A$ and $B$, respectively. $A$'s view of $(A(x_A; r_A), B(x_B; r_B))$ is the tuple $\View_A \langle A(x_A; r_A), B(x_B; r_B) \rangle = (r_A, x_A, m_{1}, m_3, \ldots )$ consisting of all the messages received by $A$ in the execution of the protocol together with the private input $x_A$ and random coins $r_{A}$. 
If we drop the random coins $r_A$ and/or $r_B$, $\View_A \langle A(x_A), B(x_B) \rangle$ becomes a random variable. $B$'s view of $(A(x_A), B(x_B))$ is defined symmetrically. 
\end{definition}
In our case, $A$ is the adversary and $B$ is the mechanism whose input is usually a database $x$. Since $A$ does not have an input in our case, we will denote the interactive protocol as $(A, B(x))$ for the ease of notation.  Since we will only be interested in $A$'s view and $A$ does not have an input, we will drop the subscript and write $A$'s view as $\view{A}{B(x)}$. 

Now we are ready to define the interactive differential privacy
as a type of interactive protocol between an adversary (without any computational limitations) and an interactive mechanism of special properties. 

\begin{definition}[Interactive Differential Privacy] 
\label{def:interactive-dp}
A randomized algorithm $\mechanism$ is an {\em $(\eps, \delta)$-differentially private interactive mechanism} if for every pair of adjacent datasets $x, x'\in \MS(\cX)$, for every adversary algorithm $\analyst$ we have:
\begin{equation} \label{req:interactive-approxdp}
\begin{split}
&\forall T \subseteq \range \left( \view{\analyst}{ \mechanism(\cdot) } \right),\\
&\Pr \left[ \view{\analyst}{\mechanism(x)}  \in T \right] \le e^\eps \Pr \left[ \view{\analyst}{\mechanism(x')} \in T \right] + \delta
\end{split}
\end{equation}
where the randomness is over the coin flips of both the algorithm $\mechanism$ and the adversary $\analyst$. 
\end{definition}
In addition to being the ``right'' modelling for many applications of differential privacy,
interactive differential privacy also captures the full power of fundamental DP mechanisms such as the  Sparse Vector Technique \cite{dwork2009complexity,roth2010interactive} and Private Multiplicative Weights \cite{hardt2010multiplicative}, which are in turn useful in the design of other DP algorithms (which can use these mechanisms as subroutines and issue adaptive queries to them).  Interactive DP was also chosen as the basic abstraction in the programming framework for the new open-source software project OpenDP~\cite{gaboardi2020programming}, which was our motivation for this research.

Despite being such a natural and useful notion, interactive DP has not been systematically studied in its own right.  It has been implicitly studied in the context of distributed forms of DP, starting with \cite{beimel2008distributed}, where the sensitive dataset is split amongst several parties, who execute a multiparty protocol to estimate a joint function of their data, while each party ensures that their portion of the dataset has the protections of DP against the other parties.  Indeed, in an $m$-party protocol, requiring DP against malicious coalitions of size $m-1$ is equivalent to requiring that each party's strategy is an interactive DP mechanism in the sense of
Definition~\ref{def:interactive-dp}.  An extreme case of this is the {\em local model} of DP, where each party holds a single data item in $\cX$ representing data about themselves~\cite{kasiviswanathan2011can}. There been extensive research about the power of interactivity in local DP; see \cite{chen2020distributed} and the references therein. 
In contrast to these distributed models, in Definition~\ref{def:interactive-dp} we are concerned with the {\em centralized DP} scenario where only one party ($\mechanism$) holds sensitive data, and how an adversarial data analyst ($\analyst$) may exploit adaptive queries to extract information about the data subjects.  

Some of the aforementioned composition theorems for noninteractive DP, such 
as in \cite{dwork2010boosting,murtagh2016complexity}, are framed in terms of an adaptive ``composition game'' where an adversary can adaptively select the mechanisms $\mechanism_0,\ldots,\mechanism_{k-1}$, and thus the resulting composition $\comp(\mechanism_0,\ldots,\mechanism_{k-1})$ can be viewed as an interactive mechanism, but the results are not framed in terms of a general definition of Interactive DP.  In particular, the mechanisms $\mechanism_i$ being composed are restricted to be noninteractive in the statements and proofs of these theorems.

\subsection{Our Contributions}

In this paper, we initiate a study of the composition of interactive DP mechanisms.  Like in the context of cryptographic protocols, there are several different forms of composition we can consider.  The simplest is {\em sequential composition}, where all of the queries to $\mechanism_{i-1}$ must be completed before any queries are issued to $\mechanism_i$.  It is straightforward to extend the proofs of the noninteractive DP composition theorems to handle sequential composition of interactive DP mechanisms; in particular the Optimal Composition Theorem (Theorem~\ref{thm:bound-opt}) extends to this case.  (Details omitted.)

Thus, we turn to {\em concurrent composition}, where an adversary can arbitrarily interleave its queries to the $k$ mechanisms.  Although the mechanisms use independent randomness, the adversary may create correlations between the executions by coordinating its actions; in particular, its
queries in one execution may also depend on messages it received in other executions.  Concurrent composability is important for the deployment of interactive DP in practice, as one or more organizations may set up multiple DP query systems on datasets that refer to some of the same individuals, and we would not want the privacy of those individuals to be violated by an adversary that can concurrently access those systems.  Concurrent composability may also be useful in the design of DP algorithms; for example, one might design a DP machine learning algorithm that uses adaptive and interleaved queries to two instantiations of an interactive DP mechanism like the Sparse Vector Technique~\cite{dwork2009complexity,roth2010interactive}.

Although the concurrent composition for the case of differential privacy has not been explored before, it has been studied extensively for many primitives in cryptography, and it is often much more subtle than the sequential composition.  (See the surveys \cite{canetti1996adaptively,goldreich2019providing}.)
\salil{Goldreich and Canetti are from https://ebooks.iospress.nl/volume/secure-multi-party-computation and Pass is from https://dl.acm.org/doi/book/10.1145/3335741}
For example, standard zero-knowledge protocols are no longer zero-knowledge when a single prover is involved in multiple, simultaneous zero-knowledge proofs with one or multiple verifiers \cite{feige1989zero,goldreich1996composition}. \salil{new citations}

\salil{for me: possibly add mention of adaptive epsilons analogy}

We use $\concomp(\mechanism_0, \dots, \mechanism_{k-1})$ to denote the concurrent composition of interactive mechanisms $\mechanism_0, \dots, \mechanism_{k-1}$.  (See Section~\ref{sec:defs} for a formal definition.)






Our first result is roughly an analogue of the Basic Composition Theorem.

\begin{restatable}{theorem}{concompbasicapprox}
\label{thm:concomp-basic-approx}
If interactive mechanisms $\mechanism_0, \dots, \mechanism_{k-1}$ are each $(\eps, \delta)$-differentially private, then their concurrent composition $\concomp(\mechanism_0, \ldots, \mechanism_{k-1})$ is
$\left(k\cdot\eps, \frac{e^{k\eps}-1}{e^\eps-1}\cdot \delta\right)$-differentially private. 

More generally, if interactive mechanism $\mechanism_i$ is $(\eps_i, \delta_i)$-differentially private for $i=0,\ldots,k-1$, then the concurrent composition $\concomp(\mechanism_0, \ldots, \mechanism_{k-1})$ is
$\left(\eps_g, \delta_g \right)$-differentially private, where 
\begin{eqnarray*}
\eps_g &=& \sum_{i=0}^{k-1} \eps_i,\text{ and}\\
\delta_g &=& \sum_{i=0}^{k-1} e^{\sum_{j=0}^{i-1}\eps_j}\cdot \delta_{i} \leq e^{\eps_g}\cdot \sum_{i=0}^{k-1} \delta_{i}.
\end{eqnarray*}
\end{restatable} 
\salil{simplified statement of thm - it was too hard to parse with the permutation - and moved the corollary into the thm statement}
\salil{make sure that these edits didn't cause problems in the technical section}

Just like in the Basic Composition Theorem for noninteractive DP (Theorem~\ref{thm:basic-comp}),
the privacy-loss parameters $\eps_i$ just sum up.  However, the bound on $\delta_g$ is worse by a factor of at most $e^{\eps_g}$.  In the typical
setting where we want to enforce a global privacy loss of $\eps_g=O(1)$, this is only a constant-factor loss compared to the Basic Composition Theorem, but that constant can be important in practice.
Note that expression for $\delta_g$ depends on the ordering of the $k$ mechanisms $\mechanism_0,\ldots,\mechanism_{k-1}$, so one can optimize it further by taking a permutation of the mechanisms that minimizes $\delta_g$.

The proof of Theorem~\ref{thm:concomp-basic-approx} is by a standard hybrid argument.  We compare the distributions of 
$H_0=\view{\analyst}{ \concomp(\mechanism_0(x),\mechanism_1(x),\ldots,\mechanism_{k-1}(x)) }$
and $H_k=\view{\analyst}{ \concomp(\mechanism_0(x'),\mechanism_1(x'),\ldots,\mechanism_{k-1}(x')) }$
on adjacent datasets $x,x'$ by changing $x$ to $x'$ for one mechanism at a time, so that $H_{i-1}$ and $H_i$ differ only on the input to $\mechanism_{i-1}$.  
To relate $H_{i-1}$ and $H_i$ we consider an adversary strategy $\adversary_i$ that emulates $\adversary$'s interaction with $\mechanism_{i-1}$, while internally simulating all of the other $\mechanism_j$'s.  Applying a ``triangle  inequality'' to the distance notion given in Requirement (\ref{req:interactive-approxdp}) yields the
result.  This proof is very similar to the proof of the ``group privacy'' property of (noninteractive) differential privacy, where $(\eps,\delta)$-DP for datasets that differ on one record implies $\left(k\cdot\eps, \frac{e^{k\eps}-1}{e^\eps-1}\cdot \delta\right)$ for datasets that differ on $k$ records.

Next we show that the Advanced and Optimal Composition Theorems (Theorems~\ref{thm:advanced-comp} and \ref{thm:bound-opt}) for noninteractive DP extend to interactive DP, provided that the mechanisms
$\mechanism_i$ being composed satisfy pure DP (i.e. $\delta_i=0$).  Note that the final composed mechanism $\concomp(\mechanism_0,\ldots,\mechanism_{k-1})$ can be approximate DP, by taking $\delta_g=\delta'>0$, and thereby allowing for a privacy loss $\eps_g$ that grows linearly in $\sqrt{k}$ rather than $k$.

We do this by extending the main proof technique of \cite{kairouz2015composition,murtagh2016complexity} to interactive DP mechanisms.  Specifically, we reduce the analysis of interactive $(\eps,0)$-DP mechanisms to that of analyzing the following simple ``randomized response'' mechanism:

\begin{definition}[\cite{warner1965randomized,dwork2006calibrating}]
For $\eps>0$, define a randomized noninteractive algorithm 
$\RR_{\eps}:\{0,1\} \rightarrow \{0,1\}$ as follows:
$$\RR_\eps(b) = 
\begin{cases} b & \mbox{w.p. $\frac{e^\eps}{1+e^\eps}$}\\
\neg b & \mbox{w.p. $\frac{1}{1+e^\eps}$.}
\end{cases}$$
\end{definition}

Note that $\RR_{\eps}$ is a noninteractive $(\eps,0)$-DP mechanism.  We show that every interactive $(\eps,0)$-DP mechanism can be, in some sense, simulated from $\RR_\eps$:
\begin{theorem}
\label{thm:reduction-pure}
Suppose that $\mechanism$ is an interactive $(\eps, 0)$-differentially private mechanism. Then for every pair of adjacent datasets $x_0, x_1$ there exists an interactive mechanism $T$ s.t. for every adversary $\analyst$ and every $b \in \{0, 1\}$ we have 
\begin{equation}
\View(\analyst, \mechanism(x_b))
    \equiv \View(\analyst, T(\RR_{\eps}(b))) \nonumber
\end{equation}
\end{theorem}
Here $T$ is an interactive mechanism that depends on $\mechanism$ as well as a fixed pair of adjacent datasets $x_0$ and $x_1$.  It receives a single bit as an output of $\RR_{\eps}(b)$, and then interacts with the adversary $\analyst$ just like $\mechanism$ would.  Kairouz, Oh, and Viswanath~\cite{kairouz2015composition} proved Theorem~\ref{thm:reduction-pure} result for the case
that $\mechanism$ and $T$ are noninteractive.  The interactive case is more involved because we need a single $T$ that works for all adversary strategies $\analyst$.  (If we allow $T$ to depend on the adversary strategy $\analyst$, then the result would readily follow from that of \cite{kairouz2015composition}, but this would not suffice for our application to concurrent composition.)

Given the Theorem \ref{thm:reduction-pure}, to analyze $\concomp(\mechanism_0(x_b),\ldots,\mechanism_{k-1}(x_b))$ on $b=0$ vs. $b=1$, it suffices to analyze $\concomp(T_0(\RR_{\eps_0}(b)),\ldots,T_{k-1}(\RR_{\eps_{k-1}}(b)))$.  An adversary's view interacting with the latter concurrent composition can be simulated entirely from the output of $\comp(\RR_{\eps_0}(b),\ldots,\RR_{\eps_{k-1}}(b))$, which is the composition of entirely noninteractive mechanisms.  Thus, we conclude:

\begin{corollary}
The Advanced and Optimal Composition Theorems (Theorems~\ref{thm:advanced-comp} and \ref{thm:bound-opt}) extend to the
concurrent composition of $(\eps_i,\delta_i)$-interactive DP mechanisms $\mechanism_i$ provided that $\delta_0=\delta_1=\cdots=\delta_{k-1}=0$.
\label{cor:extension}
\end{corollary}

We leave the question of whether or not the Advanced and/or Optimal Composition Theorems extend to the concurrent composition of approximate DP mechanisms (with $\delta_i>0$) for future work. The Optimal Composition Theorem for noninteractive approximate DP (Theorem~\ref{thm:bound-opt}) is also proven by showing that any noninteractive $(\eps,\delta)$-DP mechanism can be simulated by an approximate-DP generalization of randomized response, $\RR_{(\eps,\delta)}$, analogously to Theorem~\ref{thm:reduction-pure}.
Based on computer experiments described in Section~\ref{sec:experiments}, we conjecture that such a simulation also exists for every approximate DP interactive mechanism, and the Optimal Composition Theorem should extend at least to 2-round interactive mechanisms in which all messages are 1 bit long.

Another interesting question for future work is analyzing concurrent composition for variants of differential privacy, such as 
Concentrated DP~\cite{dwork2016concentrated,bun2016concentrated,bun2018composable}, R\'enyi DP~\cite{mironov2017renyi}, and Gaussian DP~\cite{dong2019gaussian}.  
Some of these notions require bounds on  R\'enyi divergences, e.g. that 
$$D_\alpha(\view{\analyst}{\mechanism(x)}||\view{\analyst}{\mechanism(x')})\leq \rho,$$
for adjacent datasets $x,x'$ and certain pairs $(\alpha,\rho)$.
Here sequential composition can be argued using a chain rule for R{\'e}nyi divergence:
\begin{equation}\label{eqn:chainrule}
D_\alpha((Y,Z) || (Y',Z')) \leq D_\alpha(Y||Y') + \sup_y D_\alpha(Z|_{Y=y} || Z'|_{Y'=y}).
\end{equation}
Taking $Y$ to be the view of the analyst interacting with $\mechanism_0(x)$, $Z$ to be the view of the analyst in a subsequent interaction with $\mechanism_1(x)$, and $Y'$ and $Z'$ to be analogously defined with respect to an adjacent dataset $x'$, we obtain the usual composition bound of $\rho_0+\rho_1$ on the overall R{\'e}nyi divergence of order $\alpha$, where $\rho_0$ and $\rho_1$ are the privacy-loss parameters of the individual mechanisms.  However, this argument fails for concurrent DP, since we can no longer assert the privacy properties of $\mechanism_1$ conditioned on any possible value $y$ of the adversary's view of the interaction with $\mechanism_0$. Unfortunately, the Chain Rule (\ref{eqn:chainrule}) does not hold if we replace the supremum with an expectation, so a new proof strategy is needed (if the composition theorem remains true).

\salil{Tianhao, please look at literature on interactive local DP and multiparty DP to see
if there are any results related to the post-processing conjecture in those papers.  Some papers to start with are: ``Secure multi-party DP'' by KOV, ``Extremal mechanisms for local DP'' by KOV, and ``The Role of Interactivity in Local Differential Privacy'' by Joseph et al.}
\tianhao{Extremal mechanisms for local DP by KOV \cite{kairouz2014extremal} mentions that any $\eps$-locally differentially private mechanism can be simulated from the output of a binary mechanism for sufficiently small $\eps$ (defined in sec 3.2)}
\salil{thanks - I think they are still restricting to noninteractive local DP, so no need for us to cite.  but do remind me (after the deadline) to ask around whether anything like this was known for interactive local DP.}
\salil{for me: mention open problems about concurrent composition of other forms of DP and the challenges - Renyi divergences don't satisfy the triangle inequality.}
\salil{for me: discuss experimental results in intro}

\section{Definitions and Basic Properties} \label{sec:defs}

The formal definition of the concurrent composition of interactive protocols is provided here. 

\begin{definition}[Concurrent Composition of Interactive Protocols]
Let $\mechanism_0, \dots, \mechanism_{k-1}$ be interactive mechanisms. We say $\mechanism = \concomp(\mechanism_0, \dots, \mechanism_{k-1})$ is the \emph{concurrent composition} of mechanisms $\mechanism_0, \dots, \mechanism_{k-1}$ if $\mechanism$ runs as follows:
\begin{enumerate}
    \item Random coin tosses for $\mechanism$ consist of $r=(r_0, \dots, r_{k-1})$ where $r_j$ are random coin tosses for $\mechanism_j$. 
    \item Inputs for $\mechanism$ consists of $x=(x_0, \dots, x_{k-1})$ where $x_j$ is private input for $\mechanism_j$. 
    \item $\mechanism(x, m_0, \dots, m_{i-1}; r)$ is defined as follows: 
    \begin{enumerate}
        \item Parse $m_{i-1}$ as $(q, j)$ where $q$ is a query and $j \in [k]$. If $m_{i-1}$ cannot be parsed correctly, output \halt. 
        \item Extract history $(m_0^j, \dots, m_{t-1}^j)$ from $(m_0, \dots, m_{i-1})$ where $m_i^j$ are all of the queries to mechanism $\mechanism_j$. 
        \item Output $\mechanism_j(x_j, m_0^j, \dots, m_{t-1}^j; r_j)$. 
    \end{enumerate}
\end{enumerate}
\end{definition}

We are mainly interested in the case where all mechanisms operate on the same dataset, i.e., the private input for each $\mechanism_i$ are all the same. 

We show that to prove an interactive DP mechanism is $(\eps, \delta)$-differentially private, it suffices to consider all deterministic adversaries. 

\salil{Tianhao, the following should be a lemma, not a theorem}\tianhao{done}

\begin{lemma}
An interactive mechanism $\mechanism$ is $(\eps, \delta)$-differentially private if and only if 
for every pair of adjacent datasets $x, x'$, for every \emph{deterministic} adversary algorithm $\analyst$, for every possible output set $T \subseteq \range \left( \view{\analyst}{ \mechanism(\cdot) } \right)$ 
we have 
\begin{equation}
\Pr \left[ \view{\analyst}{\mechanism(x)}  \in T \right] \le e^\eps \Pr \left[ \view{\analyst}{\mechanism(x')} \in T \right] + \delta
\label{eq:def}
\end{equation}
\end{lemma}
\begin{proof}
The necessity is immediately implied by the definition of interactive differential privacy. 
We prove the direction of sufficiency here. 
Assume that mechanism $\mechanism$ satisfies (\ref{eq:def}) for every deterministic adversary. 
Suppose, for contradiction, that there exists a randomized adversary $\analyst$ and some output set $T$ s.t. 
\begin{equation}
\Pr \left[ \view{\analyst}{\mechanism(x)}  \in T \right] > e^\eps \Pr \left[ \view{\analyst}{\mechanism(x')} \in T \right] + \delta
\label{eq:violate}
\end{equation}
Since the random coins of $\analyst$ and $\mechanism$ are independently chosen, we have 
$$
\Pr \left[ \view{\analyst}{\mechanism(x)}  \in T \right]
= \mathbb{E}_{r_A} \left[ 
\Pr_{r_\mechanism} \left[ \view{\analyst(r_A)}{\mechanism(x; r_\mechanism)}  \in T\right]
\right]. 
$$
Therefore, there must exists at least one fixed $r_A$ s.t. 
$$
\Pr \left[ \view{\analyst(r_A)}{\mechanism(x)}  \in T \right] > e^\eps \Pr \left[ \view{\analyst(r_A)}{\mechanism(x')} \in T \right] + \delta
$$
otherwise \ref{eq:violate} is impossible. 
\salil{there is still a missing step here.  the view of the randomized adversary $\analyst(r_A)$ with its coin tosses includes $r_A$ in it, whereas the view of a deterministic adversary should not have any $r_A$ in it.   formally, you should define the deterministic adversary $\analyst_{r_A}=\analyst(r_A)$ and the set $T_{r_A}=\{(m_1,m_3,\ldots) : (r_A,m_1,m_3,\ldots)\in T\}$ and argue that the probabilities are the same as above to complete the proof.  }  \tianhao{done}
Therefore, we can define a deterministic adversary $\analyst_{r_A}=\analyst(r_A)$. For set $T_{r_A}=\{(m_1,m_3,\ldots) : (r_A,m_1,m_3,\ldots)\in T\}$, since we have 
$$
\Pr \left[ \view{\analyst(r_A)}{\mechanism(x)}  \in T \right] 
= \Pr \left[ \view{\analyst_{r_A}}{\mechanism(x)}  \in T_{r_A} \right] 
$$
we know that $\analyst_{r_A}$ is a counter example for our assumption, which leads to the conclusion. 
\end{proof}

For the convenience of the proof, we introduce a variant of concurrent composition of interactive protocols, which only accept queries in the exact order of $\mechanism_0, \dots, \mechanism_{k-1}$. 

\begin{definition}[Ordered Concurrent Composition of Interactive Protocols]
Let $\mechanism_0, \dots, \mechanism_{k-1}$ be interactive mechanisms. 
We say $\mechanism = \concomp_{\order}(\mechanism_0, \dots, \mechanism_{k-1})$ is the \emph{ordered concurrent composition} of mechanisms $\mechanism_0, \dots, \mechanism_{k-1}$ if $\mechanism(x)$ runs as follows:
\begin{enumerate}
    \item Random coin tosses and inputs for $\mechanism$ are the same as $\concomp(\mechanism_0, \dots, \mechanism_{k-1})$. 
    \item $\mechanism(x, m_0, \dots, m_{i-1}; r)$ is defined as follows: 
    \begin{enumerate}
        \item Let $j = i \bmod k$, $t = \lfloor i/k\rfloor$.
        \item Output $\mechanism_j(x, m_j,m_{j+k}, \ldots, m_{j+t\cdot k};r_j)$.
    \end{enumerate}
\end{enumerate}
\end{definition}

We also introduce a special kind of interactive mechanism, which ignores all query strings begin with 0. 
\begin{definition}[Null-query Extension]
Given an interactive mechanism $\mechanism$, define its {\em null-query extension} $\mechanism^\emptyset$ defined as follows: For any input message sequence $m$, 
$\mechanism^\emptyset(x,m;r)=\mechanism(x,m';r)$ where $m'=(m'_1,\ldots,m'_k)$ such that $(1m'_1,\ldots,1m'_k)$ is the subsequence of
$m$ consisting of all strings that begin with bit 1. That is, all messages that begin with 0 are ``null queries'' that are ignored. 
By convention, $\mechanism(x,\lambda;r)=\bot$ where $\lambda$ is an empty tuple. 
\end{definition}

Now we show that in order to prove $\concomp(\mechanism_0, \dots, \mechanism_{k-1})$ is $(\eps, \delta)$-differentially private, it suffices to prove a corresponding ordered concurrent composition is also $(\eps, \delta)$-differentially private. 
We use $X \equiv Y$ to denote that two random variables $X$ and $Y$ have the same distribution. 

\begin{lemma}
$\concomp(\mechanism_0, \dots, \mechanism_{k-1})$ is an $(\eps, \delta)$-differentially private interactive mechanism if the ordered concurrent composition of the null-query extensions of $\mechanism_0, \dots, \mechanism_{k-1}$, i.e.,\\ $\concomp_{\order}(\mechanism^\emptyset_1, \dots, \mechanism^\emptyset_k)$, is an $(\eps, \delta)$-differentially private interactive mechanism.
\label{lemma:ordered}
\end{lemma}
\begin{proof}
Suppose $\concomp_{\order}\left(\mechanism^\emptyset_0, \dots, \mechanism^\emptyset_{k-1}\right)$ is $(\eps, \delta)$-differentially private. 
For every adversary $\analyst$ interacting with $\concomp \left(\mechanism_0, \dots, \mechanism_{k-1} \right)$, 
we construct another adversary $\analyst'$ interacting with $\concomp_\order(\mechanism_0^\emptyset, \dots, \mechanism_{k-1}^\emptyset)$ as follows: 
given any settings of coin tosses $r$, and any history $(q_0, a_0, \dots, q_{i-1}, a_{i-1})$ between $\analyst$ and $\concomp(\mechanism_0, \dots, \mechanism_{k-1})$,
\begin{enumerate}
    \item Let $q_i = \analyst(a_0, \dots, a_{i-1}; r)$.
    \item Parse $q_{i-1}$ as $(q_{i-1}^*, s)$ where $q_{i-1}^*$ is a query and $s \in \{0, \dots, k-1\}$ the index of target mechanism. Parse $q_{i}$ as $(q_{i}^*, t)$ in a similar way. 
    \item 
    Send the null query $0$ to $\mechanism_{(s+1)\mod k}^\emptyset, \dots, \mechanism_{(t-1)\mod k}^\emptyset$ in order. 
    \item Send $1q_i^*$ to $\mechanism_t^\emptyset$. 
\end{enumerate}
Write $\mechanism = \concomp(\mechanism_0, \dots, \mechanism_{k-1})$, and $\mechanism' = \concomp_{\order}(\mechanism_0^\emptyset, \dots, \mechanism_{k-1}^\emptyset)$. 
For every query sequence $q$ from $\analyst$, we have $\mechanism(x, q; r) = \mechanism'(x, q'; r)$ where $q'$ is the sequence of queries that $\analyst'$ asks based on $q$ (with `1' in front of every query in $q$ and additional $0$s). 
Therefore, for every $\analyst$ interact with $\mechanism$, and for every dataset $x$ we have 
$$
\view{\analyst}{\mechanism(x)} \equiv \postprocessing(
\view{\analyst'}{\mechanism'(x)})
$$
where $\postprocessing$ refers to remove all repeated answers due to the null queries. 
This immediately leads to 
\begin{align}
&\Pr[ \view{\analyst}{\mechanism(x)} \in T] \nonumber \\
&= 
\Pr[ \postprocessing( \view{\analyst'}{\mechanism'(x)}) \in T] \nonumber \\
&\le e^\eps \Pr[ \postprocessing(\view{\analyst'}{\mechanism'(x')}) \in T] + \delta \nonumber \\
&= e^\eps \Pr[ \view{\analyst}{\mechanism(x')} \in T] + \delta \nonumber
\end{align}

Therefore, $\mechanism$ is also $(\eps, \delta)$-DP. 
\end{proof}

Given Lemma \ref{lemma:ordered}, for all of the concurrent compositions we considered in this paper, we assume that the concurrent compositions are ordered. For example, if an adversary $\analyst$ is concurrently interacting with two mechanisms $\concomp(\mechanism_0, \mechanism_1)$, we assumes that the queries are alternates between $\mechanism_0$ and $\mechanism_1$.

\section{Concurrent Composition for Pure Interactive Differential Privacy}

In this section, we show that for pure differential privacy, the privacy bound for concurrent composition is the same as for sequential or noninteractive composition. 
The proof idea is that in an interactive protocol where the adversary is concurrently interacting with multiple mechanisms, its interaction with one particular mechanism could be viewed as the combination of the adversary and the remaining mechanisms interacting with that mechanism, and the differential privacy guarantee still holds for the ``combined adversary''.

\newcommand{\indis}[2]{\stackrel{( #1 , #2 )}{\approx}}

A useful notation for thinking about differential privacy and simplify presentations is defined below. 
\begin{definition}
Two random variables $Y$ and $Z$ taking values in the same output space $\mathcal{Y}$ is \emph{$(\eps, \delta)$-indistinguishable} if for every event $T \subseteq \mathcal{Y}$, we have:
$$
\Pr[Y \in T] \le e^\eps \Pr[Z \in T] + \delta
$$
$$
\Pr[Z \in T] \le e^\eps \Pr[Y \in T] + \delta
$$
which is denoted as $Y \indis{\eps}{\delta} Z$. 
\end{definition}
Notice that an algorithm $\mechanism$ is $(\eps, \delta)$ differentially private if and only if for all pairs of adjacent datasets $x, x'$, we have $\mechanism(x) \indis{\eps}{\delta} \mechanism(x')$. 

\begin{lemma}[\cite{vadhan2017complexity}]
For random variables $X, Y, Z$, if $X \indis{\eps_1}{0} Y$, $Y \indis{\eps_2}{0} Z$, then $X \indis{\eps_1+\eps_2}{0} Z$. 
\end{lemma}

\begin{restatable}[Basic Composition of Pure Interactive Differential Privacy]{theorem}{concompbasicpure}
\label{thm:concomp-basic}
If interactive mechanisms $\mechanism_0, \dots, \mechanism_{k-1}$ are each $(\eps_i, 0)$-differentially private, then their concurrent composition $\concomp(\mechanism_0, \ldots, \mechanism_{k-1})$ is $\left(\sum_{i=0}^{k-1} \eps_i, 0 \right)$-interactive differentially private. 
\end{restatable}

\begin{proof}
We first consider the simplest case that $\analyst$ concurrently interact with 2 mechanisms $\mechanism, \Tilde{\mechanism}$, and then extend the result to general amount of mechanisms. 
Suppose $\mechanism$ and $\Tilde{\mechanism}$ are each $(\eps, 0)$ and $(\Tilde{\eps}, 0)$-differentially private interactive mechanisms. 
Denote the messages received by $\analyst$ from $\mechanism$ as $(a_0, a_1, \dots, )$, and the messages received by $\analyst$ from $\Tilde{\mechanism}$ as $(\Tilde{a}_0, \Tilde{a}_1, \dots,)$. Due to Lemma \ref{lemma:ordered}, we can WLOG assume $\analyst$ alternates messages between $\mechanism$ and $\Tilde{\mechanism}$, i.e., the sequence of messages $\analyst$ received is $(a_0, \Tilde{a}_0, a_1, \Tilde{a}_1, \dots, )$. 
We use $r_\analyst$, $r_{\mechanism}, r_{\Tilde{\mechanism}}$ to denote the random coin tosses for $\analyst$, $\mechanism$, and $\Tilde{\mechanism}$, respectively. We can view $\analyst$ and $\Tilde{\mechanism}(x)$ as a single adversary $\analyst^*_{\Tilde{\mechanism}}(x)$ interacting with $\mechanism(x)$ defined as follows: 
\begin{enumerate}
    \item Random coin tosses for $\analyst^*_{\Tilde{\mechanism}}(x)$ consist of $r = (r_\analyst, r_{\Tilde{\mechanism}})$. 
    \item $\analyst^*_{\Tilde{\mechanism}}(x)(a_0, a_1, \dots, a_{i-1}; r)$ is computed as follows: 
    \begin{enumerate}
        \item $\Tilde{q}_{i-1} = \analyst(a_0, \Tilde{a}_0, 
        a_1, \Tilde{a}_1, \dots, a_{i-1}; r_\analyst)$. 
        \item $\Tilde{a}_{i-1} = \Tilde{\mechanism}(x, \Tilde{q}_0, \Tilde{q}_1, \dots, \Tilde{q}_{i-1}; r_{\Tilde{\mechanism}})$. 
        \item $q_i = \analyst(a_0, \Tilde{a}_0, \dots, a_{i-1}, \Tilde{a}_{i-1}; r_\analyst)$. 
        \item Output $q_i$.  
    \end{enumerate}
\end{enumerate}


We can see that $\analyst^*_{\Tilde{\mechanism}}(x)$ is a well-defined strategy throughout the entire interactive protocol with $\mechanism$, where the randomness of $\analyst^*_{\Tilde{\mechanism}}(x)$ is fixed as $(r_\analyst, r_{\Tilde{\mechanism}})$. 
Given a transcript of $\analyst^*_{\Tilde{\mechanism}}(x)$'s view $(r_{\analyst}, r_{\Tilde{\mechanism}}, x, a_0, a_1, \dots, )$, we can recover the corresponding transcript of $\view{\analyst}{\concomp(\mechanism(x), \Tilde{\mechanism}(x))}$ through the following post-processing algorithm $\postprocessing$, which is defined as follows: 

$\postprocessing \left(r_{\analyst}, r_{\Tilde{\mechanism}}, a_0, a_1, \dots, a_{T-1} \right)$:
\begin{enumerate}
    \item For $i = 1 \dots T-1$, compute  
    \begin{enumerate}
        \item $\Tilde{q}_{i-1} = \analyst(a_0, \Tilde{a}_0, \dots, a_{i-1}; r_{\analyst})$
        \item $\Tilde{a}_{i-1} = \Tilde{\mechanism}(x, \Tilde{q}_1, \Tilde{q}_2, \dots, \Tilde{q}_{i-1}; r_{\Tilde{\mechanism}})$
    \end{enumerate}
    \item Output $(r_{\analyst}, a_0, \Tilde{a}_0, \dots, a_{T-1}, \Tilde{a}_{T-1})$. 
\end{enumerate}

\noindent Observe that for every $\left(x, r_\analyst, r_\mechanism, r_{\Tilde{\mechanism}}\right)$, 
\begin{align}
&\postprocessing \left( \view{\analyst^*_{\Tilde{\mechanism}}(x; r_\analyst, r_{\Tilde{\mechanism}})}{\mechanism(x; r_\mechanism)} \right) \nonumber \\
&= \view{\analyst(r_\analyst)}{\concomp(\mechanism(x; r_\mechanism), \Tilde{\mechanism}(x; r_{\Tilde{\mechanism}}))} \nonumber
\end{align}
Therefore we have 
\begin{align}
    &\Pr \left[ \view{\analyst}{\concomp(\mechanism(x), \Tilde{\mechanism}(x))} \in T \right] \nonumber \\
    &\equiv 
    \Pr \left[ \postprocessing \left(\view{ \analyst^*_{\Tilde{\mechanism}}(x) }{\mechanism(x)} \right) \in T \right] \nonumber
\end{align}

for every $T \subseteq \range(\view{\analyst}{\concomp(\mechanism(x), \Tilde{\mechanism}(x))})$. 

Since $\mechanism$ is $\eps$-differentially private, we know that 
$$
\view{ \analyst^*_{\Tilde{\mechanism}}(x) }{\mechanism(x)}
\indis{\eps}{0}
\view{ \analyst^*_{\Tilde{\mechanism}}(x) }{\mechanism(x')}
$$
which leads to 
\begin{align}
&\view{\analyst}{\concomp(\mechanism(x), \Tilde{\mechanism}(x))} \nonumber \\
&\equiv 
\postprocessing \left(
\view{ \analyst^*_{\Tilde{\mechanism}}(x) }{\mechanism(x)}
\right) \nonumber \\
&\indis{\eps}{0}
\postprocessing \left(
\view{ \analyst^*_{\Tilde{\mechanism}}(x) }{\mechanism(x')}
\right) \nonumber \\
&\equiv \view{\analyst}{\concomp(\mechanism(x'), \Tilde{\mechanism}(x))} \nonumber
\end{align}

Symmetrically, we can obtain
\begin{align}
&\view{\analyst}{\concomp(\mechanism(x'), \Tilde{\mechanism}(x))} \nonumber \\
&\indis{\Tilde{\eps}}{0} 
\view{\analyst}{\concomp(\mechanism(x'), \Tilde{\mechanism}(x'))} \nonumber
\end{align}

Therefore, 
we have 
\begin{align}
&\view{\analyst}{\concomp(\mechanism(x), \Tilde{\mechanism}(x))} \nonumber \\
&\indis{\eps + \Tilde{\eps}}{0} 
\view{\analyst}{\concomp(\mechanism(x'), \Tilde{\mechanism}(x'))} \nonumber
\end{align}

The result can be easily extended to the case when more than 2 mechanisms are concurrently composed by induction.  Therefore 
for every $\eps_i \ge 0$, if interactive mechanism $\mechanism_i$ is $(\eps_i, 0)$-differentially private for $i=0, \dots, k-1$, then the concurrent composition $\concomp(\mechanism_0, \dots, \mechanism_{k-1})$ is $\left(\sum_{i=0}^{k-1} \eps_i, 0\right)$-differentially private. 
\end{proof}

This result tells us that even under concurrent composition, the privacy parameters of the resulting composed mechanisms are the ``sum up'' of the individual algorithms for the case pure differential privacy.

\section{Concurrent Composition for Approximate Interactive Differential Privacy}

In this section, we explore the privacy guarantee for the concurrent composition of interactive differential privacy when $\delta > 0$.
We show a privacy guarantee of concurrent composition in a similar logic flow as in Theorem \ref{thm:concomp-basic}, but in approximate differential privacy. 
As argued in the proof of Theorem \ref{thm:concomp-basic}, when the adversary is interacting with two mechanisms, we can view $\analyst$ and one of the mechanisms as a single adversary interacting with another mechanism, and the view of the combined adversary still enjoy the differential privacy guarantee. Therefore, if both interactive mechanisms $\mechanism$ and $\Tilde{\mechanism}$ are $(\eps, \delta)$-differentially private, then for all $S \subseteq \range(\view{\analyst}{\concomp(\mechanism(x), \Tilde{\mechanism}(x))})$, 
we know that 
\begin{align}
& \Pr \left[
\view{\analyst}{\concomp(\mechanism(x), \Tilde{\mechanism}(x))} \in S \right] \nonumber \\
& \le e^\eps \Pr \left[
\view{\analyst}{\concomp(\mechanism(x'), \Tilde{\mechanism}(x))} \in S \right] + \delta \nonumber
\end{align}
and 
\begin{align}
& \Pr \left[
\view{\analyst}{\concomp(\mechanism(x'), \Tilde{\mechanism}(x))} \in S \right] \nonumber \\
& \le e^\eps \Pr \left[
\view{\analyst}{\concomp(\mechanism(x'), \Tilde{\mechanism}(x'))} \in S \right] + \delta  \nonumber
\end{align}
and therefore we know that 
\begin{align}
& \Pr \left[
\view{\analyst}{\concomp(\mechanism(x), \Tilde{\mechanism}(x))} \in S \right] \nonumber \\
& \le e^\eps \Pr \left[
\view{\analyst}{\concomp(\mechanism(x'), \Tilde{\mechanism}(x))} \in S \right] + \delta \nonumber \\
& \le e^\eps (e^\eps \Pr \left[
\view{\analyst}{\concomp(\mechanism(x'), \Tilde{\mechanism}(x'))} \in S \right] + \delta ) + \delta \nonumber \\
& \le e^{2\eps} \Pr \left[
\view{\analyst}{\concomp(\mechanism(x'), \Tilde{\mechanism}(x'))} \in S \right] 
+ (1+e^\eps) \delta \nonumber
\end{align}

A more general concurrent composition bound is stated and derived as follows:
\begin{theorem}[Theorem \ref{thm:concomp-basic-approx} restated]
Let $\sigma: \{0, 1, \ldots, n-1\} \rightarrow \{0, 1, \ldots, n-1\}$ be any permutation of $0, \ldots, n-1$. 
If interactive mechanisms $\mechanism_0, \dots, \mechanism_{k-1}$ are each $(\eps_i, \delta_i)$-differentially private, then their concurrent composition $\concomp(\mechanism_0, \ldots, \mechanism_{k-1})$ is
$\left(\sum_{i=0}^{k-1} \eps_i, \delta_g \right)$-differentially private, where 
$$
\delta_g = \min_{\sigma} 
\left( 
\delta_{\sigma(0)} + \sum_{i=1}^{k-1} e^{\sum_{j=0}^{i-1}\eps_{\sigma(j)}} \delta_{\sigma(i)} 
\right)
$$
For mathematical convenience, we use an upper bound for $\delta_g$ in practice and $\concomp(\mechanism_0, \ldots, \mechanism_{k-1})$ is $\left( \sum_{i=0}^{k-1} \eps_i, k e^{\sum_{i=0}^{k-1} \eps_i} \max_i(\delta_i) \right)$-differentially private.
\end{theorem}
\begin{proof}
We use a hybrid argument. For each $0 \le i \le k-1$, since $\mechanism_i$ is $(\eps_i, \delta_i)$ differentially private, we know that 
\begin{align}
& \Pr \left[
\view{\analyst}{\concomp(\mechanism_0(x'), \ldots,  \mechanism_{i-1}(x'), \mechanism_{i}(x), \ldots, \mechanism_{k-1}(x))} \in S \right] \nonumber \\
& \le e^{\eps_i}
\Pr \left[
\view{\analyst}{\concomp(\mechanism_0(x'), \ldots,  \mechanism_{i-1}(x'), \mechanism_{i}(x'), \ldots, \mechanism_{k-1}(x))} \in S \right] + \delta_i \nonumber
\end{align}
by viewing $\analyst$ and $\mathcal{M}_0, \dots, \mathcal{M}_{i-1}, \mathcal{M}_{i+1},  \mathcal{M}_{k-1}$ as a combined adversary and follow a similar argument as in the proof of Theorem \ref{thm:concomp-basic-approx}. 

Therefore, 
\begin{align}
& \Pr \left[
\view{\analyst}{\concomp(\mechanism_0(x), \mechanism_1(x), \ldots, \mechanism_{k-1}(x))} \in S \right] \nonumber \\
& \le e^{\eps_0}
\Pr \left[
\view{\analyst}{\concomp(\mechanism_0(x'), \mechanism_1(x), \ldots, \mechanism_{k-1}(x))} \in S \right] + \delta_0 \nonumber \\
& \le e^{\eps_0}
(e^{\eps_1}
\Pr \left[
\view{\analyst}{\concomp(\mechanism_0(x'), \mechanism_1(x'), \ldots, \mechanism_{k-1}(x))} \in S \right] + \delta_1) + \delta_0 \nonumber \\
& \le \ldots \nonumber \\
& \le e^{\sum_{i=0}^{k-1} \eps_i} 
\Pr \left[
\view{\analyst}{\concomp(\mechanism_0(x'), \mechanism_1(x'), \ldots, \mechanism_{k-1}(x'))} \in S \right] \nonumber \\
&~~~~+ (\delta_0+e^{\eps_0}\delta_1  + e^{\eps_0+\eps_1} \delta_2 + \ldots + e^{\sum_{i=0}^{k-2} \eps_i}\delta_{k-1} ) \nonumber 
\end{align}
We can see that the $\delta$ term of $\concomp(\mechanism_0, \ldots, \mechanism_{k-1})$ depends on different permutations of $(\mechanism_0, \ldots, \mechanism_{k-1})$, and the tightest possible bound for the $\delta$ term is 
$$
\min_{\sigma} 
\left( 
\delta_{\sigma(0)} + \sum_{i=1}^{k-1} e^{\sum_{j=0}^{i-1}\eps_{\sigma(j)}} \delta_{\sigma(i)} 
\right)
$$
We also note that $\delta_0+e^{\eps_0}\delta_1  + e^{\eps_0+\eps_1} \delta_2 + \ldots + e^{\sum_{i=0}^{k-2} \eps_i}\delta_{k-1} \le k e^{\sum_{i=0}^{k-1} \eps_i} \max_i (\delta_i)$, which is more easier to work with in practice. 
\end{proof}

Notice that if the privacy parameters are homogeneous, i.e. every interactive mechanism is $(\eps, \delta)$ differentially private, then this bound reduce to the bound of group privacy for $(\eps, \delta)$-differential privacy.

\section{Characterization of $\concomp$ for Pure Interactive Differential Privacy}


\cite{kairouz2015composition} shows that to analyze the composition of arbitrary noninteractive $(\eps_i, \delta_i)$-DP algorithms, it suffices to analyze the composition of the following simple variant of randomized response. 

\newcommand{\iamzero}{\texttt{`Iam0'}}
\newcommand{\iamone}{\texttt{`Iam1'}}

\begin{definition}[\cite{kairouz2015composition}]
Define a randomized noninteractive algorithm 
$\RR_{(\eps, \delta)}:\{0,1\} \rightarrow \{0,1, \iamzero , \iamone \}$ as follows:
$$
\begin{array}{ll}
\operatorname{Pr}\left[\RR_{(\eps, \delta)}(0)=\iamzero \right]=\delta & \operatorname{Pr}\left[\RR_{(\eps, \delta)}(1)=\iamzero \right]=0 \\
\operatorname{Pr}\left[\RR_{(\eps, \delta)}(0)=0\right]= (1-\delta) \cdot \frac{e^{\eps}}{1+e^{\eps}} & \operatorname{Pr}\left[\RR_{(\eps, \delta)}(1)=0\right]= (1-\delta) \cdot \frac{1}{1+e^{\eps}} \\
\operatorname{Pr}\left[\RR_{(\eps, \delta)}(0)=1\right]= (1-\delta) \cdot \frac{1}{1+e^{\eps}} & \operatorname{Pr}\left[\RR_{(\eps, \delta)}(1)=1\right]= (1-\delta) \cdot \frac{e^{\eps}}{1+e^{\eps}} \\
\operatorname{Pr}\left[\RR_{(\eps, \delta)}(0)=\iamone \right]=0 & \operatorname{Pr}\left[\RR_{(\eps, \delta)}(1)=\iamone \right]=\delta
\end{array}
$$
\end{definition}
Note that $\RR_{(\eps, \delta)}$ is a noninteractive $(\eps, \delta)$-differentially private mechanism. \cite{kairouz2015composition} and \cite{murtagh2016complexity} showed that $\RR_{(\eps, \delta)}$ can be used to simulate the output of every (noninteractive) $(\eps, \delta)$-DP algorithm on adjacent databases. $\RR$ refers to ``randomized response", as this mechanism is a generalization of the classic randomized response to $\delta>0$ and $\eps\neq \ln 2$  \cite{warner1965randomized}. \salil{edited this sentence}

\begin{theorem}[\cite{kairouz2015composition}]
\label{thm:noninterac}
Suppose that $\mechanism$ is $(\eps, \delta)$-differentially private. Then for every pair of adjacent datasets $x_0, x_1$ there exists a randomized algorithm $T$ s.t. $T(\RR (b))$ is identically distributed to $\mechanism(x_b)$ for both $b=0$ and $b=1$. 
\end{theorem}
This theorem is useful due to one of the central properties of differential privacy is that it is preserved under ``post-processing'' \cite{dwork2006calibrating,dwork2014algorithmic}, which is formulated as follows:
\begin{lemma}[Post-processing]
\label{thm:post-processing}
If a randomized algorithm $\mechanism: \X \rightarrow \Y$ is $(\eps, \delta)$-differentially private, and $\mathcal{F}: \Y \rightarrow \mathcal{Z}$ is any randomized function, then $\mathcal{F} \circ \mechanism: \X \rightarrow \mathcal{Z}$ is also $(\eps, \delta)$-differentially private. 
\end{lemma}

In noninteractive setting, Theorem \ref{thm:noninterac} can be used to prove the optimal composition theorem \cite{kairouz2015composition,murtagh2016complexity} since to analyze the composition of arbitrary $(\eps_i, \delta_i)$-DP algorithms, it suffices to analyze the composition of $\RR_{(\eps_i, \delta_i)}$ algorithms. 

If we are able to prove a similar result that arbitrary interactive differential private mechanisms can also be simulated by the post-processing of randomized response where the interactive post-processing algorithm does not depend on the adversary, then we will be able to extend all results of composition theorem for noninteractive mechanisms to interactive mechanisms. In this paper, we consider the case of pure differential privacy. 


\begin{theorem}[Theorem \ref{thm:reduction-pure} restated]
Suppose that $\mechanism$ is an interactive $(\eps, 0)$-differentially private mechanism. Then for every pair of adjacent datasets $x_0, x_1$ there exists an interactive mechanism $T$ s.t. for every adversary $\analyst$ and every $b \in \{0, 1\}$ we have 
\begin{equation}
\View(\analyst, \mechanism(x_b))
    \equiv \View(\analyst, T(\RR_{(\eps, 0)}(b))) \nonumber
\end{equation}
\end{theorem}

\begin{proof}
For arbitrary sequence of queries $\bm{q}^{(t)} = (q_0, \dots, q_{t-1})$ from $\analyst$, 
we denote by
$
\vec{\mechanism}(x, \bm{q}^{(t)}) = 
(\mechanism(x, \bm{q}^{(1)}), \mechanism(x, \bm{q}^{(2)}), \dots, \mechanism(x, \bm{q}^{(t)}))
$
the random variable consisting the first $t$ responses from mechanism $\mechanism$. 
We construct the interactive mechanism $T$ receiving queries $\bm{q}^{(t)}$ as follows: 
\begin{enumerate}
    \item If $t=0$, we have
    \begin{equation}
    \begin{split}
    \Pr \left[ T(0, q_0)=a_0 \right] = \frac{ e^\eps \Pr[\mechanism(x_0, q_0)=a_0] - \Pr[\mechanism(x_1, q_0)=a_0] }{e^\eps - 1}
    \end{split}
    \end{equation}
    \begin{equation}
    \begin{split}
    \Pr \left[ T(1, q_0)=a_0 \right] = \frac{ e^\eps \Pr[\mechanism(x_1, q_0)=a_0] - \Pr[\mechanism(x_0, q_0)=a_0] }{e^\eps - 1}
    \end{split}
    \end{equation}
    \item If $t>0$, given earlier responses $(a_0, \dots, a_{t-2})$, we define 
    \begin{equation}
    \begin{split}
    &\Pr \left[ T(0, \bm{q}^{(t)}) = a_{t-1} | a_0, \dots, a_{t-2} \right] \\
    &= \frac{e^\eps \Pr \left[ \vec{\mechanism}(x_0, \bm{q}^{(t)})=(a_0, \dots, a_{t-1}) \right] - \Pr \left[ \vec{\mechanism}(x_1, \bm{q}^{(t)})=(a_0, \dots, a_{t-1}) \right] }{(e^\eps - 1) \Pr \left[ \vec{T}(0, \bm{q}^{(t-1)}) = (a_0, \dots, a_{t-2}) \right] }
    \label{eq:dist-post}
    \end{split}
    \end{equation}
    \begin{equation}
    \begin{split}
    &\Pr \left[ T(1, \bm{q}^{(t)}) = a_{t-1} | a_0, \dots, a_{t-2} \right] \\
    &= \frac{e^\eps \Pr \left[\vec{\mechanism}(x_1, \bm{q}^{(t)})=(a_0, \dots, a_{t-1}) \right] - \Pr \left[\vec{\mechanism}(x_0, \bm{q}^{(t)})=(a_0, \dots, a_{t-1})\right] }{(e^\eps - 1) \Pr \left[\vec{T}(1, \bm{q}^{(t-1)}) = (a_0, \dots, a_{t-2})\right] }
    \end{split}
    \end{equation}
\end{enumerate}

Therefore, the distribution of $\vec{T}$ is 
\begin{align}
    &\Pr \left[ \vec{T}(0, \bm{q}^{(t)}) = (a_0, \dots, a_{t-1}) \right] \nonumber \\
    &= \frac{e^\eps \Pr \left[\vec{\mechanism}(x_0, \bm{q}^{(t)})=(a_0, \dots, a_{t-1}) \right] - \Pr \left[\vec{\mechanism}(x_1, \bm{q}^{(t)})=(a_0, \dots, a_{t-1}) \right] }{e^\eps - 1} \nonumber
\end{align}
\begin{align}
    &\Pr \left[ \vec{T}(1, \bm{q}^{(t)}) = (a_0, \dots, a_{t-1}) \right] \nonumber \\
    &= \frac{e^\eps \Pr \left[\vec{\mechanism}(x_1, \bm{q}^{(t)})=(a_0, \dots, a_{t-1})\right] - \Pr \left[\vec{\mechanism}(x_0, \bm{q}^{(t)})=(a_0, \dots, a_{t-1}) \right] }{e^\eps - 1} \nonumber
\end{align}

We can easily verify that all of the above are valid probability distributions. 
For example, 
\begin{align}
    &\sum_{a_{t-1}} \Pr \left[ T(0, \bm{q}^{(t)}) = a_{t-1} | a_0, \dots, a_{t-2} \right] \nonumber \\
    &= \frac{e^\eps \sum_{a_{t-1}} \Pr \left[ \vec{\mechanism}(x_0, \bm{q}^{(t)})=(a_0, \dots, a_{t-1}) \right] - \sum_{a_{t-1}} \Pr \left[ \vec{\mechanism}(x_1, \bm{q}^{(t)})=(a_0, \dots, a_{t-1}) \right] }{(e^\eps - 1) \Pr \left[ \vec{T}(0, \bm{q}^{(t-1)}) = (a_0, \dots, a_{t-2}) \right] } \nonumber \\
    & = \frac{e^\eps \Pr \left[ \vec{\mechanism}(x_0, \bm{q}^{(t)})=(a_0, \dots, a_{t-2}) \right] -  \Pr \left[ \vec{\mechanism}(x_1, \bm{q}^{(t)})=(a_0, \dots, a_{t-2}) \right] }{(e^\eps - 1) \Pr \left[ \vec{T}(0, \bm{q}^{(t-1)}) = (a_0, \dots, a_{t-2}) \right] } \nonumber \\
    & = 1
\end{align}
and for every possible $a_{t-1}$, the probability density is never negative since 
$$
\Pr \left[ \vec{\mechanism}(x_0, \bm{q}^{(t)})=(a_0, \dots, a_{t-1}) \right]
\le e^\eps \Pr \left[ \vec{\mechanism}(x_1, \bm{q}^{(t)})=(a_0, \dots, a_{t-1}) \right]
$$
as $\mechanism$ is $(\eps, 0)$-DP. 


We now show 
$$
\View(\analyst, \mechanism(x_b))
    \equiv \View(\analyst, T(\RR_{(\eps, 0)}(b)))
$$
for the case of $b=0$. 

Fix any possible view $(r, a_0, \dots, a_{t-1})$, we can derive the queries $\bm{q}^{(t)} = (q_0, \ldots, q_{t-1})$ from $\analyst$, where $q_i = \analyst(a_0, \ldots, a_{i-1}; r)$. Denote $R$ as the random variable of the randomness of $\analyst$. 
\begin{align}
&\Pr \left[ \View(\analyst, T(\RR_{(\eps, 0)}(0))) = (r, a_0, \dots, a_{t-1}) \right] \nonumber \\
& = \Pr \left[\RR_{(\eps, 0)}(0)=0 \right] \Pr \left[ \View(\analyst, T(0)) = (r, a_0, \dots, a_{t-1}) \right] \nonumber \\
&~~~~+ \Pr \left[\RR_{(\eps, 0)}(1)=0 \right] \Pr \left[ \View(\analyst, T(1)) = (r, a_0, \dots, a_{t-1}) \right] \nonumber \\
& = \frac{e^\eps}{1+e^\eps} \Pr \left[ \View(\analyst, T(0)) = (r, a_0, \dots, a_{t-1}) \right] \nonumber \\
&~~~~+ \frac{1}{1+e^\eps} \Pr \left[ \View(\analyst, T(1)) = (r, a_0, \dots, a_{t-1}) \right] \nonumber \\
& = \frac{e^\eps}{1+e^\eps} \Pr \left[ R=r \right]
\Pr \left[ \vec{T}(0, \bm{q}^{(t)}_r) = (a_0, \dots, a_{t-1})  | R=r \right] \nonumber \\
&~~~~+ \frac{1}{1+e^\eps} \Pr \left[ R=r \right] 
\Pr \left[ \vec{T}(1, \bm{q}^{(t)}_r) = (a_0, \dots, a_{t-1}) | R=r \right] \nonumber \\
& = \Pr \left[ R=r \right] \Pr \left[\vec{\mechanism}(x_0, \bm{q}^{(t)}_r)=(a_0, \dots, a_{t-1}) | R=r \right]
 \nonumber \\
& = \Pr \left[\View(\analyst, \mechanism(x_0)) = (r, a_0, \dots, a_{t-1}) \right] \nonumber
\end{align}
The case of $b=1$ could be similarly proved. Therefore, we proved the existence of such an interactive mechanism $T$ for any $(\eps, 0)$ interactive DP mechanisms. 
\end{proof}

\newcommand{\concomppriv}{\mathrm{OptCompPriv}}
\newcommand{\privloss}{\mathrm{PrivLoss}}

The above theorem suggests that the noninteractive $\RR_{\left(\varepsilon, 0 \right)}$ can simulate any $(\eps, 0)$ interactive DP algorithm. Since it is known that post-processing preserves differential privacy (Lemma \ref{thm:post-processing}), it follows that to analyze the concurrent composition of arbitrary $(\eps_i, 0)$ interactive differentially private algorithms, it suffices to analyze the composition of randomized response $\RR_{\left(\eps_i, 0 \right)}$. 
For an interactive mechanism $\mechanism$, we define $\privloss(\mechanism, \delta) = \inf \left\{\varepsilon \geq 0: \text{$\mechanism$ is $(\eps,\delta)$-DP}\right\}$, 
thus given a target security parameter $\delta_g$, the privacy loss of the concurrent composition of mechanisms $\mechanism_0, \ldots, \mechanism_{k-1}$ is denoted as $\privloss(\concomp(\mechanism_{0}, \ldots, \mechanism_{k-1}),\delta_g)$. 
\salil{better notation: for an interactive mechanism $\mechanism$, define $\mathrm{PrivLoss}(\mechanism,\delta) = \inf \left\{\varepsilon \geq 0: \text{$\mechanism$ is $(\eps,\delta)$-DP}\right\}$, and then you can write things like $\mathrm{PrivLoss}(\concomp(\mechanism_{0}, \ldots, \mechanism_{k-1}),\delta)$}\tianhao{done}
\salil{for a future version: for noninteractive mechanisms, it would be nice to have the notation Comp defined instead of just ConComp to emphasize that there is no issue of concurrency}
\salil{new sentences.  also write `noninteractive' rather than `non0interactive throughout'} \tianhao{done}
When the mechanisms $\mechanism_i$ are noninteractive (like $\RR_{(\eps,\delta)}$) we write $\comp$ rather than $\concomp$. 

\begin{lemma} 
Suppose there are interactive mechanisms $\mechanism_0, \ldots, \mechanism_{k-1}$ where for each $0 \le i \le k-1$, $\mechanism_i$ is $(\eps_i, 0)$-differentially private. For any values of $\eps_0, \dots, \eps_{k-1} \ge 0$, $\delta_g \in [0, 1)$, we have
\begin{align}
&\privloss(\concomp(\mechanism_{0}, \ldots, \mechanism_{k-1}),\delta_g) \nonumber \\
&= \privloss \left( \comp(\RR_{(\eps_0, 0)}, \dots, \RR_{(\eps_{k-1}, 0)}), \delta_g \right) \nonumber
\end{align}
\end{lemma}
\begin{proof}
We want to show that 
\begin{align}
&\inf \left\{\varepsilon_{g} \geq 0: \concomp(\mechanism_{0}, \ldots, \mechanism_{k-1}) \text { is }\left(\varepsilon_{g}, \delta_{g}\right)\mathrm{-DP}\right\} \nonumber \\
&= \inf \left\{\varepsilon_{g} \geq 0: \comp \left(\RR_{(\eps_0, 0)}, \dots, \RR_{(\eps_{k-1}, 0)} \right) \text { is }\left(\varepsilon_{g}, \delta_{g}\right)\mathrm{-DP}\right\} \nonumber
\end{align}

Since the noninteractive $\RR_{(\eps_0, 0)}, \dots, \RR_{(\eps_{k-1}, 0)}$ can be viewed as a special case of interactive DP mechanisms, we have 
\begin{align}
&\inf \left\{\varepsilon_{g} \geq 0: \concomp(\mechanism_{0}, \ldots, \mechanism_{k-1}) \text { is }\left(\varepsilon_{g}, \delta_{g}\right)\mathrm{-DP}\right\} \nonumber \\
&\ge \inf \left\{\varepsilon_{g} \geq 0: \comp \left(\RR_{(\eps_0, 0)}, \dots, \RR_{(\eps_{k-1}, 0)} \right) \text { is }\left(\varepsilon_{g}, \delta_{g}\right)\mathrm{-DP}\right\} \nonumber
\end{align}

For the other direction, suppose $\comp \left(\RR_{(\eps_0, 0)}, \dots, \RR_{(\eps_{k-1}, 0)} \right)$ is $(\eps_g^*, \delta_g)$-DP. 
By post-processing inequality, we know any for any tuple of post-processing interactive mechanisms $T_0, \ldots, T_{k-1}$, $\concomp \left(T_0 \left(\RR_{(\eps_0, 0)} \right), \ldots, T_{k-1} \left(\RR_{(\eps_{k-1}, 0)} \right) \right)$ is also $(\eps_g^*, \delta_g)$-DP. 
We know from Theorem \ref{thm:reduction-pure} that for every pair of adjacent datasets $x_0, x_1$, there must exist interactive mechanisms $T_0, \ldots, T_{k-1}$ such that for every adversary $\analyst$, $\view{\analyst}{\mechanism_i(x_b)}$ is identically distributed as $\view{\analyst}{T_i(\RR_{ \left( \varepsilon, 0 \right)}(b))}$ for all $i = 0, \ldots, k-1$. 
Therefore, we know that $\concomp(\mechanism_{0}, \ldots, \mechanism_{k-1})$ is also $(\eps_g^*, \delta_g)$-DP. 
Taking the infimum over $\eps_g^*$ will then complete the proof.

\salil{there are still missing steps here.  I suggest removing the inf's from the statement, as you've already said how you are going to do the infs in an earlier paragraph. You are assuming here that the Comp of the RR's is $(\eps_g^*,\delta_g)$-DP.  The first equality above needs to be justified by arguing that for every adversary interacting with ConComp of the $\mechanism_i$'s, its view is identical to its view interacting with ConComp of the $T_i(\RR)$'s.  Then the next inequality needs to be justified by explaining how the latter view is a postprocessing of Comp of the $\RR$'s.}\tianhao{if we remove inf, does the inequality for set makes sense? }
\salil{It's not just removing the inf's.  Above you set up the proof as follows: ``suppose $\comp \left(\RR_{(\eps_0, 0)}, \dots, \RR_{(\eps_{k-1}, 0)} \right)$ is $(\eps_g^*, \delta_g)$-DP. We will show that $\concomp(\mechanism_0, \ldots, \mechanism_{k-1})$ is also $(\eps_g^*, \delta_g)$-DP.''  So the proof should be a series of implications ($\Rightarrow$), not a series of inequalities or set inclusions: $\comp(\RR)$ is DP $\Rightarrow$ $\concomp(T(\RR))$ is DP $\Rightarrow$ $\concomp(\mechanism)$ is DP}\tianhao{done}

\end{proof}

We note that $\RR_{(\eps_0, 0)}, \dots, \RR_{(\eps_{k-1}, 0)}$ are noninteractive mechanisms, therefore we can use any composition theorems for noninteractive DP mechanisms to bound the privacy parameter of their composition. The tightest composition theorem for noninteractive DP is derived in \cite{murtagh2016complexity}. 

\begin{theorem}[Optimal Composition Theorem for noninteractive DP]
\label{thm:nonint-opt}
If $\mechanism_0, \dots, \mechanism_{k-1}$ are each $(\eps_i, \delta_i)$-differentially private, then given the target security parameter $\delta_g$, the privacy parameter of concurrent composition $\concomp(\mechanism_0, \dots, \mechanism_{k-1})$ is
upper bounded by the least value of $\varepsilon_{g} \geq 0$ such that
$$
\frac{1}{\prod_{i=0}^{k-1}\left(1+\mathrm{e}^{\varepsilon_{i}}\right)} 
\sum_{S \subseteq\{0, \ldots, k-1\}} \max 
\left
\{\mathrm{e}^{\sum_{i \in S} \varepsilon_{i}}-\mathrm{e}^{\varepsilon_{g}} \cdot \mathrm{e}^{\sum_{i \notin S} \varepsilon_{i}}, 0
\right\} 
\leq 1-\frac{1-\delta_{g}}{\prod_{i=0}^{k-1}\left(1-
\delta_{i} \right)}
$$
\end{theorem}
Therefore, we are ready to bound the concurrent composition for an arbitrary set of interactive differentially private algorithms by simply plugging parameters to the optimal composition bound for noninteractive DP mechanisms in \cite{murtagh2016complexity}. 

\salil{I am going to change the theorems that are in the intro, so for now it's safest to explicitly write each theorem statement in the technical sections rather than automatically duplicating the intro statements}\tianhao{done}

\begin{theorem}[Corollary \ref{cor:extension} Restated]
\label{thm:bound-opt-restate}
If $\mechanism_0, \dots, \mechanism_{k-1}$ are each $(\eps_i, 0)$-differentially private, then given the target security parameter $\delta_g$, the privacy parameter of concurrent composition $\concomp(\mechanism_0, \dots, \mechanism_{k-1})$ is
upper bounded by the least value of $\varepsilon_{g} \geq 0$ such that
$$
\frac{1}{\prod_{i=0}^{k-1}\left(1+\mathrm{e}^{\varepsilon_{i}}\right)} 
\sum_{S \subseteq\{0, \ldots, k-1\}} \max 
\left
\{\mathrm{e}^{\sum_{i \in S} \varepsilon_{i}}-\mathrm{e}^{\varepsilon_{g}} \cdot \mathrm{e}^{\sum_{i \notin S} \varepsilon_{i}}, 0
\right\} 
\leq \delta_g
$$
A special case when all $\mechanism_0, \dots, \mechanism_{k-1}$ are $(\eps, 0)$-differentially private, then privacy parameter is upper bounded by the least value of $\eps_g \ge 0$ such that
$$
\frac{1}{\left(1+\mathrm{e}^{\varepsilon}\right)^k} 
\sum_{i=0}^k {k \choose i} 
\max 
\left
\{\mathrm{e}^{i \eps }-\mathrm{e}^{\varepsilon_{g}} \cdot \mathrm{e}^{(k-i) \varepsilon}, 0
\right\} 
\leq \delta_g
$$
\end{theorem}

\section{Experimental Results}
\label{sec:experiments}

\salil{edited this par}
In this section, we present empirical evidence for our conjecture that the Optimal Composition Theorems can be extended to the concurrent composition of approximate DP mechanisms. 
\salil{I'll add this conjecture to intro}
Specifically, we experimentally evaluate the conjecture for 3-message interactive mechanisms with 1-bit messages, as illustrated in Figure \ref{fig:two-step-mechanism}. 
The input for the mechanism is a bit $x \in \{0, 1\}$ (corresponding to fixing two adjacent datasets). 
In the first round, the mechanism outputs a bit $a_0$ regardless of the query, so we omit $q_0$ and directly writing the probability of outputting $a_0$ as $\Pr[\mechanism(x)=a_0]$. In the second round, the mechanism receives a query bit $\analyst(a_0)$ from the adversary, and output another bit $a_1$. Each such mechanism $\mechanism_{\bm{p}}$ is defined by 10 parameters $\bm{p} = \left(p_0, p_{00}, p_{01}, p_{10}, p_{11}, p_0', p_{00}', p_{01}', p_{10}', p_{11}' \right)$, where $p_0 = \Pr[\mechanism_{\bm{p}}(0)=0]$, $p_0' = \Pr[\mechanism_{\bm{p}}(1)=0]$, $p_{ij} = \Pr[\mechanism_{\bm{p}}(0, j)=(i, 0)]$, $p_{ij}' = \Pr[\mechanism_{\bm{p}}(1, j)=(i, 0)]$. 
\salil{added subscript $\bm{p}$}
We note that the concurrent composition of two copies of such a mechanism already has a nontrivial interleaving, as shown in Figure~\ref{fig:concomp}.

\salil{dropped paragraph on concurrent composition simulation, since those results are subsumed by the post-processing results}

We experimentally test whether instantiations of this 2-round interactive mechanism that are $(\eps, \delta)$-DP can be simulated as the interactive post-processing of randomized response $\RR_{(\eps,\delta)}$.   Specifically, we sample over 10,000 choices of the parameter vector $\bm{p}$ defining the mechanism $\mechanism_{\bm{p}}$.  For each one, we pre-define a value for $\delta$ and compute $\eps = \privloss(\mechanism_{\bm{p}},\delta)$ through enumerating over all possible adversaries. 
\salil{Tianhao, pelase explain how you calculated $\eps$ and $\delta$ here} \tianhao{done}

Next, we used linear programming to see if there exists an interactive post-processing mechanism $\vec{T}$ which takes an output from $\RR_{(\eps,\delta)}$, and sets it to have the exact same output distribution as the original 2-round for every possible query $q = (q_1)$ and output sequence $(a_0, a_1)$:
\salil{following equations are taking up way too much space. use eqnarray* with lefteqn to get a better formatting}\tianhao{the best I can do ...} \salil{how about the first example below?  note also the added cdots}
\tianhao{done}
\begin{eqnarray*}
\lefteqn{\Pr \left[\vec{\mechanism}(0, q) = (a_0, a_1)\right] = \delta\cdot \Pr \left[\vec{T}(\iamzero, q) = (a_0, a_1) \right]}\\
&& + (1-\delta)\cdot \frac{e^\eps}{e^\eps+1} \Pr \left[\vec{T}(0, q) = (a_0, a_1)\right]
+ (1-\delta)\cdot \frac{1}{e^\eps+1} \Pr \left[\vec{T}(1, q) = (a_0, a_1) \right]
\end{eqnarray*}

\begin{eqnarray*}
\lefteqn{\Pr \left[\vec{\mechanism}(1, q) = (a_0, a_1) \right] = (1-\delta) \cdot \frac{1}{e^\eps+1} \Pr \left[\vec{T}(0, q) = (a_0, a_1) \right]} \\
&& + (1-\delta) \cdot \frac{e^\eps}{e^\eps+1} \Pr \left[\vec{T}(1, q) = (a_0, a_1)\right] 
+ \delta \cdot \Pr \left[\vec{T}(\iamone, q) = (a_0, a_1) \right]
\end{eqnarray*}

Each $\Pr \left[\vec{T}(c, q) = (a_0, a_1) \right]$ is an unknown parameter here, where $c \in \{0,1, \iamzero , \iamone \}$. 
We also enforce them formulating valid distributions:
$$
\forall c, q, a_0, a_1, \Pr \left[\vec{T}(c, q) = (a_0, a_1) \right] \ge 0
$$
$$
\forall c, \analyst, \sum_{a_0, a_1} \Pr \left[\vec{T}(c, \analyst(a_0)) = (a_0, a_1) \right] = 1
$$
Besides, to construct a valid two-round mechanism, the probability of outputting $a_0$ in the first round should not depend on the future query $q_1$:
$$
\forall c, a_0, \sum_{a_1} \Pr \left[\vec{T}(c, 0) = (a_0, a_1) \right] = \sum_{a_1} \Pr \left[\vec{T}(c, 1) = (a_0, a_1) \right]
$$

\salil{write out these constraints here!}\tianhao{done}
We use the linear programming solver from SciPy \cite{virtanen2020scipy} for solving the linear equation systems. 

In all of our trials, we find a feasible $\vec{T}$, concluding that each of the mechanisms
$\mechanism_{\bm{p}}$ can be simulated by the post-processing of randomized response of the same $(\eps, \delta)$ parameters. 

Based on the above findings, we conjecture that the concurrent composition of interactive DP mechanisms may still have the same bound as the composition for noninteractive DP mechanisms. Besides, we might be able to prove it through a similar construction of interactive post-processing mechanisms as we did in Theorem \ref{thm:reduction-pure}. This means that every interactive DP mechanisms can be reduced to noninteractive randomized response. We leave the resolution of these conjectures for future work.

\salil{if you have time before the deadline, it would be nice to extend the experiments to larger message spaces, and possibly even try 3-round protocols, and briefly mention the results here.  Ideally, I'd like to have a bit more evidence before asserting a conjecture}\tianhao{Unfortunately I'm working on 3 NIPS papers ... and two of them are first authored by me ... work until 5am almost everyday since the last week ... having saying that, I'm very serious about this paper! }
\salil{no problem - I understand!}

\salil{for our future work (post-deadline): these experiments suggest an approach to proving the post-processing theorem by linear programming.  suppose that we have a mechanism that cannot be simulated by randomized response.  Then the LP above is not feasible.  Hence by strong LP duality, the dual LP has a feasible solution.  Maybe we can interpret the solution to the dual LP as an adversary violating the $(\eps,\delta)$-DP property of $\mechanism$.} \tianhao{I see!}


\begin{figure}[H]
    \centering
	\includegraphics[width=0.7\columnwidth]{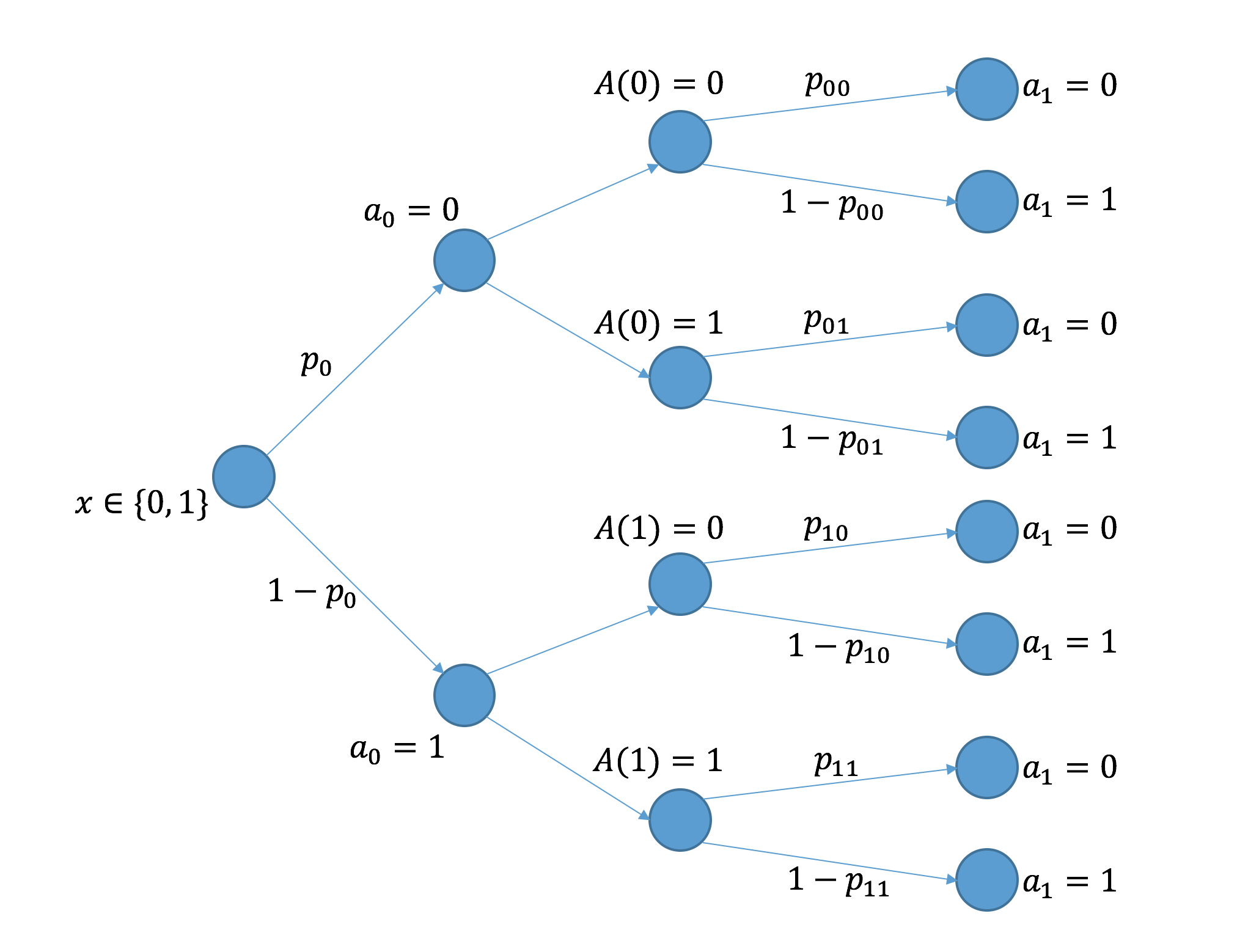}
	\caption{2-round mechanism we use in the experiment.}
	\label{fig:two-step-mechanism}
\end{figure}

\begin{figure}[H]
    \centering
	\includegraphics[width=0.7\columnwidth]{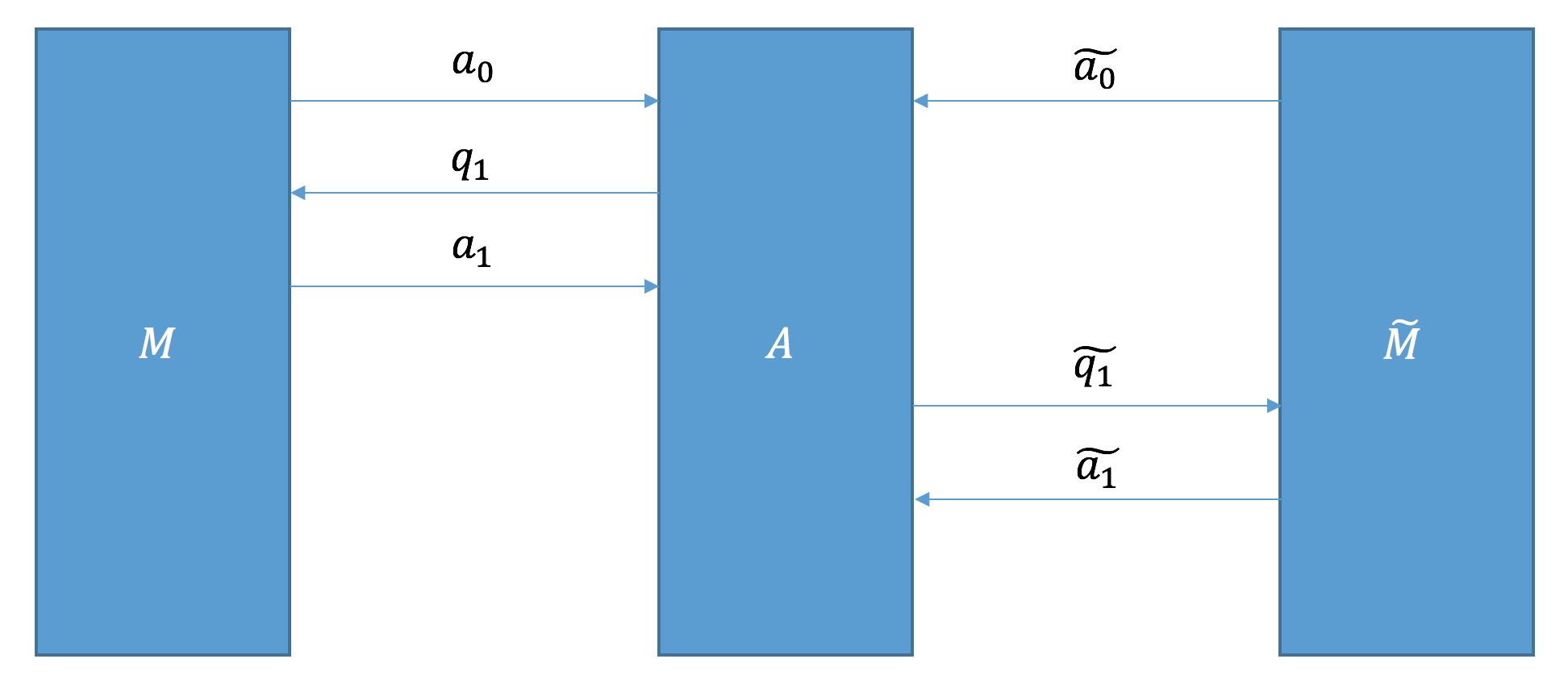}
	\caption{Concurrent Composition of 2-round Mechanisms}
	\label{fig:concomp}
\end{figure}

\salil{Tianhao, please add an acknowledgments section, thanking your committee members and Cynthia Dwork for helpful comments and discussions}\tianhao{done}

\section{Acknowledgement}
We sincerely thank Boaz Barak, Cynthia Dwork, Marco Gaboardi, Michael Hay, Weiwei Pan, and Andy Vyrros for helpful comments and discussions.

\newpage

\bibliographystyle{splncs04}
\bibliography{references}
\end{document}